\documentclass[twoside,reqno]{amsart}

\usepackage[square,comma]{natbib}
\usepackage[english]{babel} 
\usepackage{exscale}   
\usepackage{amsmath}
\usepackage{amsfonts}
\usepackage{amssymb}
\usepackage{amsxtra}
\usepackage{stmaryrd}
\usepackage{xspace}
\usepackage{epsfig}
\usepackage{mathrsfs}
\usepackage{pstricks}
\usepackage{enumitem}
\usepackage{latexsym}
\usepackage{graphics}
\usepackage{xypic}

\newtheorem{theorem}{Theorem}[section]
\newtheorem{proposition}[theorem]{Proposition}
\newtheorem{corollary}[theorem]{Corollary}

\newtheorem{definition}[theorem]{Definition}

\newtheorem{remark}[theorem]{Remark}

\newcommand{\BH}{\mbox{$\mathcal{B}(\mathcal{H})$}}

\newcommand{\BL}{\mbox{$\mathcal{B}(\mathcal{L})$}}
\newcommand{\BP}{\mbox{$\mathcal{B}(\mathcal{P})$}}

\newcommand{\OH}{\mbox{$\Omega^{\mathcal{H}}$}}
\newcommand{\OK}{\mbox{$\Omega^{\mathcal{K}}$}}
\newcommand{\OP}{\mbox{$\Omega^{\mathcal{P}}$}}
\newcommand{\OKi}{\mbox{$\Omega_\infty^{\mathcal{K}}$}}
\newcommand{\OPi}{\mbox{$\Omega_\infty^{\mathcal{P}}$}}

\newcommand{\Cset}{\mathbb{C}}

\newcommand{\Nset}{\mathbb{N}}

\newcommand{\Zset}{\mathbb{Z}} 
\newcommand{\CA}{\ensuremath{{\mathcal A}}\xspace}         
\newcommand{\CC}{\ensuremath{{\mathcal C}}\xspace}         
\newcommand{\CE}{\ensuremath{{\mathcal E}}\xspace}         
\newcommand{\CG}{\ensuremath{{\mathcal G}}\xspace}         
\newcommand{\CH}{\ensuremath{{\mathcal H}}\xspace}         
\newcommand{\CK}{\ensuremath{{\mathcal K}}\xspace}         
\newcommand{\CL}{\ensuremath{{\mathcal L}}\xspace}         
\newcommand{\CP}{\ensuremath{{\mathcal P}}\xspace}         
\newcommand{\CU}{\ensuremath{{\mathcal U}}\xspace}         
\newcommand{\CY}{\ensuremath{{\mathcal Y}}\xspace}
\newcommand{\eins}{\ensuremath{{\rm 1\kern-.25em l}}\xspace}  
 
\begin{document}

\title[Non-commutative Markov Chains and Multi-analytic Operators]
{Non-commutative Markov Chains and Multi-analytic Operators}
\author{Rolf Gohm}
\address{Dr.Rolf Gohm \\
Institute of Mathematics and Physics \\ 
Aberystwyth University \\
Aberystwyth, 
Ceredigion SY23 3BZ, UK}
\email{rog@aber.ac.uk}
\subjclass[2000]{47A13, 46L53, 47A48, 47A40, 81R15}
\keywords{repeated interaction, open quantum system, non-commutative Markov chain, multivariate operator theory, row contraction, outgoing Cuntz scattering system, transfer function, multi-analytic operator, input-output formalism, linear system, quantum control, observability, scattering theory, characteristic function}
\date{February 19, 2009}
\begin{abstract}
We study a model of repeated interaction between quantum systems which can be thought of as a non-commutative Markov chain. It is shown that there exists an outgoing Cuntz scattering system associated to this model which induces an input-output formalism with a transfer function corresponding to a multi-analytic operator, in the sense of multivariate operator theory. Finally we show that observability for this system is closely related to the scattering theory of non-commutative Markov chains. 
\end{abstract}
\maketitle

\section{Introduction}
\label{section:intro}

The purpose of this paper is to point to an interesting connection between non-commutative Markov chains, which are a well known mathematical model for open quantum systems embedded into an environment (see for example \cite{Ku03} for a recent survey), 
and multi-analytic operators, which are a central idea in the development of multivariate operator theory (as described for example in \cite{Po95}). Though some hints about such connections between non-commutative probability theory and multivariate operator theory can be found in the literature, see for example \cite{Bh96,Bh01,Go04,DG07},
this is not widely appreciated and a more systematic investigation is missing. We want to present an argument that both fields can benefit from each other and make a start by studying a rather elementary mathematical toy model from this point of view. Let us describe some more background in order to clarify the developments and applications we have in mind. 

The theory of open quantum systems embedded into an environment has been driven more and more to the language of linear systems theory, input-output formalisms and control. We only mention
the theoretical physics survey in \cite{GZ00} and the recent investigations on quantum networks \cite{GJ}. The use of Laplace transforms and analytic functions is routine in classical linear systems and control theory but it is not at all obvious how to generalize this part of the theory to the quantum world. Though a notion of transfer function for bosonic fields has been recently developed \cite{YK03a,YK03b}, see also \cite{GJ,GGY08}, we do not follow this approach here but instead we suggest to introduce analyticity into the theory of quantum systems by the use of multi-analytic operators.

These operators, also known as analytic intertwining operators in \cite{BV05}, have been introduced and studied in \cite{Po89a,Po89b,Po95} in the development of a multi-variable version for row contractions of the Sz.-Nagy/Foias-theory \cite{SF70} and it has been shown that many aspects of the classical theory of analytic functions generalize to this setting, see for example \cite{Po06}. Among them are techniques which are relevant to linear systems and control.

It is therefore clear that the study of quantum systems would benefit if it is possible to construct such operators as a generalized type of transfer function for these systems. We show in this paper that if we can model the quantum system and the interaction in a certain way by a non-commutative Markov chain then this is indeed possible.

As this seems to be a new idea and as our aims in this paper are partly expository, to convince theoretical physicists working on interactions of quantum systems and pure mathematicians working on multivariate operator theory that there exists a promising intersection of their interests, we have made no attempt to study the most general model but instead we investigate a very elementary discrete time toy model of repeated interactions between an open system and its environment. Though stripped to its bare basics it nevertheless already exhibits some interesting features which we expect to be typical. The model is defined in a purely mathematical way and then interpreted physically in Section \ref{section:interaction}. In Section
\ref{section:cuntz} we shift our attention to multivariate operator theory and prove that we can find an outgoing Cuntz scattering system in the sense of \cite{BV05} inside our model of repeated interaction.
It is well known that this leads to a transfer function which in fact corresponds to a multi-analytic operator. We work out in Section
\ref{section:transfer} how this transfer function and related concepts can be developed in a convenient way for our setting and we give a physical interpretation in terms of certain experimental records. Another concept motivated from a control point of view is observability. We show in Section \ref{section:observ} that in our setting observability is closely related to a scattering theory for non-commutative Markov chains introduced in \cite{KM00} whose further development \cite{Go04,GKL06} in fact led the author into the direction followed in this paper. We refer to the end of Section \ref{section:observ} for a more detailed discussion of this scattering interpretation which gives further evidence that our transfer function via multivariate operator theory is also a natural tool for the study of the physical system. We further comment there about some closely related investigations about characteristic functions of ergodic tuples and of liftings in \cite{DG07,DG}.

\section{A model for repeated interactions} \label{section:interaction}

We want to study a very elementary mathematical model for an interaction between quantum systems. The model is specified by a few operator theoretic data.

\begin{definition} \label{def:interaction}
Given \\
\begin{itemize}
\item[]
three Hilbert spaces $\CH,\, \CK,\, \CP$ \\
\item[]
a unitary operator $U: \CH \otimes \CK \rightarrow \CH \otimes \CP \;$
(i.e. $U^* U = U U^* = \eins$) \\
\item[]
unit vectors $\OH \in \CH,\, \OK \in \CK,\, \OP \in \CP$ such that
\[
U\,\big(\OH \otimes \OK\big) = \OH \otimes \OP
\]
\end{itemize}
we say that $U$ is an {\it interaction with vacuum vectors} $\OH, \OK, \OP$.
\end{definition}

We can form the infinite tensor products $\CK_\infty := \bigotimes^\infty_1 \CK$ and $\CP_\infty := \bigotimes^\infty_1 \CP$ with distinguished unit vectors
$\OKi = \bigotimes^\infty_1 \OK$ and $\OPi = \bigotimes^\infty_1 \OP$,
see \cite{KR83}, 11.5.29.
We denote the $\ell$-th copies by $\CK_\ell$ and $\Omega_\ell^\CK$ and use a natural notation built upon it.
For example, for $\Omega_1^\CK \otimes \Omega_2^\CK$ we write $\Omega_{[1,2]}^\CK$,
for $\CK_3 \otimes \CK_4 \otimes \CK_5$ we write 
$\CK_{[3,5]}$ and identify it with
$\Omega_{[1,2]}^\CK \otimes \CK_{[3,5]} \otimes 
\Omega_{[6,\infty)}^\CK \subset \CK_\infty$.
Similarly for $\CP$ where we have spaces $\CP_{[m,n]}$
and vectors $\Omega_{[m,n]}^\CP$.
Embeddings without further explanations are always understood in such a way, using the vacuum vectors, for example
\[
\CH \simeq \CH \otimes \OKi \subset \CH \otimes \CK_\infty \supset
\OH \otimes \CK_\infty \simeq \CK_\infty .
\]

We can now define repeated interactions. For $\ell \in \Nset$ let
\[
U_\ell: \CH \otimes \CK_\infty \rightarrow \CH \otimes \CK_{[1,\ell-1]} \otimes \CP_\ell \otimes \CK_{[\ell+1,\infty)}
\]
be the unitary operator which is equal to $U$ on $\CH \otimes \CK_\ell$ and which acts identically on the other factors of the tensor product. Then the repeated interaction up to time $n \in \Nset$ is defined by
\[
U(n):= U_n \ldots U_1: \CH \otimes \CK_\infty \rightarrow \CH \otimes \CP_{[1,n]} \otimes \CK_{[n+1,\infty)}
\]

A physical interpretation of a very similar scheme is developed in detail in \cite{BHJ}. Let us quickly sketch the basic ideas. To do that assume for the moment that $\CK = \CP$ and $\OK = \OP$ and that $U$ is a unitary operator on $\CH \otimes \CK$.
We may think of $\CH$ as the (quantum-mechanical) Hilbert space of an atom,
$\CK$ as the Hilbert space of a (portion of a) light beam and $U$ as a
(toy model of their) quantum mechanical interaction. After the interaction the (pure) state $\eta \in \CH \otimes \CK$ of the coupled system is changed into $U \eta \in \CH \otimes \CK$ (Schr\"{o}dinger picture). The only thing we require from our vacuum vectors $\OH$ and $\OK$ is that their tensor product represents a state which is not affected by the interaction $U$.

If we further assume that at time $1$ an interaction happens for $\CH \otimes \CK_1$, after the interaction this portion of the light beam goes away from the atom forever but another portion of the light beam described by $\CK_2$
reaches the atom and interacts at time $2$, etc., then we obtain a repeated interaction and $U(n)= U_n \ldots U_1$ is the Schr\"{o}dinger dynamics up to time $n$. In fact, in \cite{BHJ} there is a detailed derivation of the Heisenberg dynamics $X \mapsto U(n)^*\, X\, U(n)$ for observables $X$ which confirms our sketchy interpretation above. It is shown in \cite{BHJ} that this toy model of a quantum mechanical atom-field interaction allows a discussion of many of the basic equations and constructions of quantum filtering and quantum control which is parallel to the usual continuous time treatment. Moreover the continuous time results can be obtained from that as a scaling limit. Repeated interactions of a similar type are also investigated in \cite{BG08}.

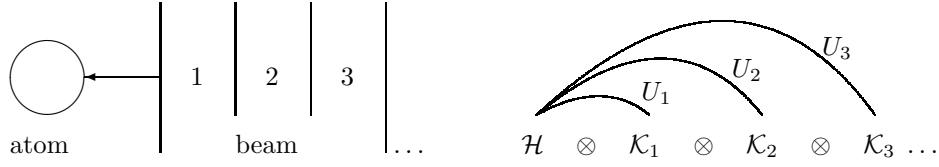
\setlength{\unitlength}{1cm}
\begin{picture}(15,3.5)
\put(0.5,2){\circle{1}} 
\put(2.0,2){\vector(-1,0){1}}
\put(2.0,1){\line(0,1){2}}
\put(3.0,1.5){\line(0,1){1.5}}
\put(4.0,1.5){\line(0,1){1.5}}
\put(5.0,1){\line(0,1){2}}
\put(5.1,1){\ldots}
\put(2.4,1.9){$1$}
\put(3.4,1.9){$2$}
\put(4.4,1.9){$3$}
\put(0,1){atom}
\put(3.0,1){beam}
\put(6.8,1){$\CH \quad \otimes \quad \CK_1 \quad \otimes \quad \CK_2 \quad \otimes \quad \CK_3 \; \ldots$}
\qbezier(7,1.5)(7.9,2)(8.5,1.5)
\qbezier(7,1.5)(8.8,3)(10,1.5)
\qbezier(7,1.5)(9.7,4)(11.5,1.5)
\put(8.4,1.7){$U_1$}
\put(9.6,2){$U_2$}
\put(10.8,2.3){$U_3$}
\end{picture}

\begin{remark} \normalfont
The reader who wants to see the propagator (unitary time evolution operator) for this repeated interaction in a more explicit way should consult the literature about non-commutative Markov chains and their coupling structure, see for example \cite{Ku03} for a survey. For our interactions this gives a description which is essentially equivalent to the model presented here but it includes additionally the dynamics of the beam outside the range of the interaction as a tensor shift. More precisely, consider the enlarged space $\CH \otimes \bigotimes^\infty_{-\infty} \CK$ and let $S$ be the tensor shift to the left (on the copies of $\CK$ while acting identically on $\CH$).
Then $\big( U(n) \big)_{n\in\Nset_0}$ 
(with $U(0) := \eins$) 
can be thought of as a right unitary cocycle for $S$ , i.e. $U(n+m) = S^{-n} U(m) S^n U(n)$, and we obtain an evolution by a unitary group $\big( \tilde{U}^n \big)_{n\in\Zset}$ such that $\tilde{U}^n = S^n U(n)$
for $n\in\Nset_0$. The unitary operator $\tilde{U}$ is the propagator of the coupled atom-field system which includes both the interactions and the movement of the beam to the left before and after the interactions. 

Hence from this enlarged point of view our model of repeated interaction is actually what physicists would call an interaction picture of the dynamics (but we shall nevertheless, with a slight abuse of language, continue to call it Schr\"{o}dinger resp. Heisenberg picture if we change states resp. observables using $U(n)$ in the original model).

While the enlarged model is useful for the study of the complete propagator and for other structural questions we will not use it in this paper because we are more interested in the input-output map of the system and for this purpose the original model of repeated interaction presented before is simpler and more direct. So from now on we work with the model of interaction specified in Definition \ref{def:interaction} and the one-sided model of repeated interaction built from it. 
\end{remark}

We can think of our model as a non-commutative Markov chain or, from a physicist's point of view, as a Markovian approximation of a repeated atom-field interaction.
Let us elaborate a bit on the Markovianity.
In the Heisenberg picture the change of an observable $X \in \BH$ until time $n$ compressed to $\CH$ is given by
\[
Z_n(X) = P_{\CH}\, U(n)^*\; X \otimes 1 \;U(n) |_{\CH},
\]
where $P_{\CH}$ denotes the orthogonal projection from 
$\CH \otimes \CK_\infty$ onto $\CH$. Let $\big( \epsilon_j \big)$
be an orthonormal basis of the Hilbert space $\CP$. For $\xi \in \CH$
we get
\[
U(\xi \otimes \OK) = \sum_j A_j\,\xi \otimes \epsilon_j
\]
with operators $A_j \in \BH$. Then a short computation yields
\[
Z_n(X) = \sum_{j_1,j_2,\ldots,j_n} A^*_{j_1} \ldots A^*_{j_n}
\,X\,A_{j_n} \ldots A_{j_1} = Z^n(X),
\]
where $Z = \sum_j A^*_j \cdot A_j: \BH \rightarrow \BH$ is a unital completely positive map called the transition operator of the non-commutative Markov chain. The semigroup property of the compressed dynamics established above is one of the basic features of Markovianity which for classical Markov chains is expressed by the Chapman-Kolmogorov equations, see for example \cite{Fe68}. 

Our condition about vacuum vectors yields
\[
\langle \OH, X\, \OH \rangle \;=\; \langle \OH, Z(X)\, \OH \rangle\;,
\]
i.e., the vector state induced by $\OH$ is invariant for $Z$.

\section{Outgoing Cuntz scattering systems} \label{section:cuntz}

To develop a different approach to the model in the previous section we review some notions from multivariable operator theory. 
See for example \cite{Po89a,Po89b,BGM06}.

Suppose $T_1,\ldots,T_d \in \BL$ for a Hilbert space $\CL$. We allow $d=\infty$ but simplify our notation by pretending that $d$ is finite and leave in the following the suitable modifications and limits to the reader. Then

$\underline{T} = (T_1,\ldots,T_d)$ is called a {\it row contraction} if it is contractive as an operator from $\bigoplus^d_1 \CL$ to $\CL$ or,
equivalently, if $\sum^d_1 T_j T^*_j \le 1$.

$\underline{T} = (T_1,\ldots,T_d)$ is called a {\it row isometry} if
it is isometric as an operator from $\bigoplus^d_1 \CL$ to $\CL$ or,
equivalently, if the $T_j$ are isometries with orthogonal ranges.

$\underline{T} = (T_1,\ldots,T_d)$ is called a {\it row unitary} if
it is unitary (isometric and surjective) as an operator from $\bigoplus^d_1 \CL$ to $\CL$ or,
equivalently, if the orthogonal ranges of the isometries $T_j$ 
together span $\CL$.

A row isometry $\underline{T} = (T_1,\ldots,T_d)$ is called a {\it row shift} if there exists a subspace $\CE$ of $\CL$ (the wandering subspace) such that $\CL = \bigoplus_{\alpha \in F^+_d} T_\alpha \CE$. Here (as always in this paper) $\bigoplus$ denotes an orthogonal sum,
$F^+_d$ is the free semigroup with generators $1,\ldots,d$.
If $\alpha \in F^+_d$ is the word $\alpha_\ell \ldots \alpha_1$
of length $|\alpha| = \ell$
(where $\alpha_j \in \{1,\ldots,d\}$) then $T_\alpha = T_{\alpha_\ell} \ldots T_{\alpha_1}$. For the empty word $\emptyset$ we define
$|\emptyset|=0$ and $T_\emptyset = \eins$. 

Recall the following definition from \cite{BV05}, Chapter 5. 
An {\it outgoing Cuntz scattering system} is a collection
\[
\big( \CL,\, \underline{V} = (V_1,\ldots,V_d),\, \CG^+_*,\, \CG \big)
\]
where $\underline{V}$ is a row isometry on the Hilbert space
$\CL$ and $\CG^+_*$ and $\CG$ are subspaces of $\CL$ such that
\begin{enumerate}
\item
$\CG^+_*$ is the smallest $\underline{V}$-invariant subspace containing
\[
\CE_* := \CL \ominus span_{j=1,\ldots,d} \;V_j \CL\;,
\]
thus $\underline{V} |_{\CG^+_*}$ is a row shift and 
$\CG^+_* = \bigoplus_{\alpha \in F^+_d} V_\alpha \CE_*$
\item
$\underline{V} |_{\CG}$ is a row shift, thus
$\CG = \bigoplus_{\alpha \in F^+_d} V_\alpha \CE$ with
\[
\CE := \CG \ominus span_{j=1,\ldots,d} \;V_j \CG.
\]
\end{enumerate}

\begin{remark} \normalfont
Note that $\CG^+_*$ and $\CE_*$ are determined by $\CL$ and $\underline{V}$ and hence in principle could be omitted from the defining data. But the idea behind this concept (the `scattering') is to study the relative position of the two row shifts. 
The decomposition
$\CL = (\CL \ominus \CG^+_*) \oplus \CG^+_*$ gives the Wold decomposition of $\underline{V}$ as a row unitary plus a row shift, see \cite{Po89a}, Theorem 1.3, which is a multi-variable version of the classical Wold decomposition of an isometry. 
\end{remark}

Now we associate an outgoing Cuntz scattering system to an interaction $U$ with vacuum vectors $\OH,\,\OK,\,\OP$, as specified in Definition \ref{def:interaction}. We define the data of the scattering system in terms of those data as follows:
\[
\CL := (\CH \otimes \CK_\infty)^o :=
(\CH \otimes \CK_\infty) \ominus \Cset (\OH \otimes \OKi),
\]
For $\xi \otimes \eta \in \CH \otimes \CK_\infty$ and an orthonormal basis $\big( \epsilon_j \big)_{j=1,\ldots,d}$ of $\CP$
\[
V_j \big( \xi \otimes \eta \big) := U^*(\xi \otimes \epsilon_j) \otimes \eta
\in (\CH \otimes \CK_1) \otimes \CK_{[2,\infty)}
\]
which by linear extension defines $\underline{V} = (V_1,\ldots,V_d)$
on $(\CH \otimes \CK_\infty)^o$ (details in Theorem \ref{thm:cuntz}).
As discussed above, $\CE_*$
and $\CG^+_* = \bigoplus_{\alpha \in F^+_d} V_\alpha \CE_*$
are implicitly defined as the shift part of $\underline{V}$.
Finally
\[
\CE := \CH \otimes (\Omega_1^\CK)^\perp \otimes \Omega_{[2,\infty)}^\CK,
\quad  
\CG = \bigoplus_{\alpha \in F^+_d} V_\alpha \CE.
\]
Note that $\underline{V} = (V_1,\ldots,V_d)$ depends on the choice of the orthonormal basis
$\big( \epsilon_j \big)_{j=1,\ldots,d}$ of $\CP$ but a change of basis by a unitary $d \times d$-matrix has the same effect on the $V_j$ and leads to a closely related system which for many purposes can be identified with the original one.

For technical reasons we also define $\underline{\widehat{V}}$ as the extension of $\underline{V}$ to $\CH \otimes \CK_\infty$ given by the same formula, in other words we have
\[
\widehat{V}_j \big( \OH \otimes \OKi \big) 
= U^* \,(\OH \otimes \epsilon_j) \otimes \Omega_{[2,\infty)}^\CK
\]

We need some auxiliary operators.
By $Q_{[1,n]}$ we denote the orthogonal projection from $\CP_\infty$
to $\CP_{[1,n]} \subset \CP_\infty$. Further we define
$Q_n: \CH \otimes \CP_{[1,n]} \otimes \CK_{[n+1,\infty)} \rightarrow
\CP_{[1,n]} \subset \CP_\infty$ by
\[
Q_n\,(\xi \otimes \zeta \otimes \rho) \;=\; 
\langle \OH, \xi \rangle \; \zeta \; \langle \Omega_{[n+1,\infty)}^\CK, \rho \rangle
\]
which also maps onto $\CP_{[1,n]}$.

\begin{proposition} \label{prop:W}
There exists a coisometry $\widehat{W}: \CH \otimes \CK_\infty \rightarrow \CP_\infty$ such that
\begin{eqnarray*}
\widehat{W}\, &=& sot-\lim_{n\to\infty} Q_n\,U(n) ,
\\
\,\widehat{W}^* &=& sot-\lim_{n\to\infty} U(n)^* |_{\CP_\infty} ,
\end{eqnarray*}
where $sot$ stands for `strong operator topology'.
\\ 
By restriction we obtain a coisometry 
$W: (\CH \otimes \CK_\infty)^o \rightarrow 
(\CP_\infty)^o := \CP_\infty \ominus \Cset \OPi$.
\end{proposition}

\begin{proof}
We start by constructing the adjoint $\widehat{W}^*$. We have
$U^*_\ell |_{\CP_\ell}: \CP_\ell \simeq \OH \otimes \CP_\ell \rightarrow \CH \otimes \CK_\ell$ and $\lim_{n\to\infty} U_1^* \ldots U_n^* \zeta$
clearly exists for $\zeta \in \bigcup_{N \ge 0} \CP_{[0,N]}$. In fact, because of the invariance property of the vacuum vectors
only finitely many of the $U^*_\ell$ act nontrivially on $\zeta$.
Because the $U^*_\ell$ are isometries we obtain an isometric extension $\widehat{W}^*$ to the closure $\CP_\infty$. Its adjoint is a coisometry
$\widehat{W}: \CH \otimes \CK_\infty \rightarrow \CP_\infty$. More explicitly, for $\eta \in \CH \otimes \CK_{[1,m]}$ and $\zeta \in \CP_{[1,n]}$ and $m \le n$ we obtain
\[
\langle\, \widehat{W} \eta,\, \zeta \,\rangle
= \langle\, \eta,\, \widehat{W}^* \zeta \,\rangle
= \langle\, Q_n\, U(n) \eta,\, \zeta \,\rangle 
\]
We conclude that
\[
Q_{[1,n]} \widehat{W} \eta = Q_n U(n) \eta
\]
For $n\to\infty$ the left hand side converges to $\widehat{W} \eta$
and we have
\[
\widehat{W} \eta = \lim_{n\to\infty} Q_n U(n) \eta
\]
Finally we can extend this formula from $\bigcup_{m \ge 1} \CH \otimes \CK_{[1,m]}$ to the whole of $\CH \otimes \CK_\infty$ by
continuity.
The restriction $W$ acts as shown because 
$\widehat{W} (\OH \otimes \OKi) = \OPi$.
\end{proof}

\begin{theorem} \label{thm:cuntz}
Let $U$ be an interaction with vacuum vectors $\OH,\,\OK,\,\OP$, as specified in Definition \ref{def:interaction}.
Then with the definitions above
\[
\big( (\CH \otimes \CK_\infty)^o,\, \underline{V} = (V_1,\ldots,V_d),\, \CG^+_*,\, \CG \big)
\]
is an outgoing Cuntz scattering system.
We obtain an explicit formula for $\CE_*$ as follows.
\[
\CE_* = W^* \CY \subset \CH \otimes \CK_1
\]
with
\[
\CY := \OH \otimes (\Omega_1^\CP)^\perp \otimes \Omega_{[2,\infty)} \subset \OH \otimes \CP_\infty \simeq \CP_\infty
\]
Finally we have
\[
(\CH \otimes \CK_\infty)^o = \CH^o  \oplus \CG\,,
\]
where $\CH^o := \CH \ominus \Cset \OH$. 
\end{theorem}

\begin{proof}

The $\widehat{V}_j$ are isometries with orthogonal ranges because the $\epsilon_j$ form an orthonormal set and $U^*$ is isometric.  Note that because the $\epsilon_j$ form a basis, $\underline{\widehat{V}}$ is even a row unitary, i.e., 
$span_{j=1,\ldots,d}\; \widehat{V}_j (\CH \otimes \CK_\infty) = \CH \otimes \CK_\infty$.

Now suppose that $\sum_i \xi_i \otimes \eta_i \in (\CH \otimes \CK_\infty)^o$ with
$\xi_i \in \CH$ and $\eta_i \in \CK_\infty$. Further suppose
that $\zeta = \OH \otimes \zeta_1^\CP \otimes \Omega_{[2,\infty)}^\CP
\in \OH \otimes \CP_\infty \simeq \CP_\infty$. Then for all $j$
\begin{eqnarray*}
&& \langle\, V_j \sum_i \xi_i \otimes \eta_i,\,
W^* \,\zeta \,\rangle
= \langle\, \sum_i U^* (\xi_i \otimes \epsilon_j) \otimes \eta_i,
\, \lim_{n\to\infty} U^*_1\ldots U^*_n\; \OH \otimes \zeta_1^\CP \otimes \Omega_{[2,\infty)}^\CP \,\rangle \\
&=& \langle\, \sum_i \xi_i \otimes \epsilon_j \otimes \eta_i,
\, \OH \otimes \zeta_1^\CP \otimes \Omega_{[2,\infty)}^\CK \,\rangle = 0
\end{eqnarray*}
because $\sum_i \xi_i \otimes \eta_i \perp \OH \otimes \OKi$.

By choosing $\zeta_1^\CP = \Omega_1^\CP$ we conclude that $(\CH \otimes \CK_\infty)^o$ is invariant for all $V_j$ and hence $\underline{V}$ is a row isometry on $(\CH \otimes \CK_\infty)^o$. 

Clearly $W^* \CY \subset \CH \otimes \CK_1$ and $W^* \CY \perp
\OH \otimes \OKi$, hence also
$W^* \CY \subset (\CH \otimes \CK_\infty)^o$. 

By choosing $\zeta_1^\CP \perp \Omega_1^\CP$,
i.e. $\OH \otimes \zeta_1^\CP \otimes \Omega_{[2,\infty)} \in \CY$, we conclude that $W^* \CY$ is orthogonal to $span_{j=1,\ldots,d}\; V_j (\CH \otimes \CK_\infty)^o$.
To prove that $W^* \CY$ is equal to the wandering subspace $\CE_*$ of $\underline{V}$ it remains to be shown that no other vectors are orthogonal to $span_{j=1,\ldots,d}\; V_j (\CH \otimes \CK_\infty)^o$. In fact, suppose $\eta \in (\CH \otimes \CK_\infty)^o$ is orthogonal not only to $span_{j=1,\ldots,d}\; V_j (\CH \otimes \CK_\infty)^o = span_{j=1,\ldots,d}\; \widehat{V}_j (\CH \otimes \CK_\infty)^o$ but also to $W^* \CY$. The latter (together with $\eta \in (\CH \otimes \CK_\infty)^o$) means that for all 
$\zeta_1^\CP \in \CP_1$
\[
\eta \perp U^* (\OH \otimes \zeta_1^\CP) \otimes \Omega_{[2,\infty)}^\CK
\] 
which gives $\eta \perp span_{j=1,\ldots,d}\; \widehat{V}_j
(\OH \otimes \OKi)$. Combined with $\eta \perp span_{j=1,\ldots,d}\, \widehat{V}_j (\CH \otimes \CK_\infty)^o$
this implies
\[
\eta \perp span_{j=1,\ldots,d}\, \widehat{V}_j (\CH \otimes \CK_\infty)
= \CH \otimes \CK_\infty\;.
\]
The last equality follows from the fact that
$\underline{\widehat{V}}$ is a row unitary. We conclude that $\eta = 0$. 

Further it is clear that 
$\CE = \CH \otimes (\Omega_1^\CK)^\perp \otimes \Omega_{[2,\infty)}^\CK$
is contained in $(\CH \otimes \CK_\infty)^o$. 
Now we prove that $V_\alpha \CE$ and $V_\beta \CE$ are orthogonal to each other if $\alpha \not= \beta$ in $F^+_d$. If $|\alpha| = |\beta|$ but
$\alpha \not= \beta$
then even the ranges of $V_\alpha$ and $V_\beta$ are orthogonal (because $\underline{V}$ is a row isometry). But if, say,  $|\alpha| > |\beta|$ then $V_\alpha \CE \perp V_\beta \CE$ follows by considering the form of $\CE$ (consider the inner product at the tensor factor $\CK_{|\alpha|+1}$).
Finally, to see that $\CG = \bigoplus_{\alpha \in F^+_d} V_\alpha \CE
= (\CH \otimes \CK_\infty)^o \ominus \CH^o$ check by induction that
for all $n \in \Nset$
\[
\CH \otimes \CK_{[1,n]} \;=\;
\big( \CH \otimes \OKi \big) 
\oplus \big( \CH \otimes (\Omega_1^\CK)^\perp \otimes \Omega_{[2,\infty)}^\CK \big)
\oplus \big( \CH \otimes \CK_1 \otimes (\Omega_2^\CK)^\perp \otimes \Omega_{[3,\infty)}^\CK \big) \oplus \ldots
\]
\[
\ldots 
\oplus \big( \CH \otimes \CK_{[1,n-1]} \otimes (\Omega_n^\CK)^\perp \otimes \Omega_{[n+1,\infty)}^\CK \big)
\]
\[
= \big( \Cset\,\OH \oplus \CH^o \big)
\oplus \CE \oplus \bigoplus^d_{j=1} V_j \CE \oplus
\ldots 
\oplus \bigoplus_{|\alpha|=n-1} V_\alpha \CE\;.
\]
Restricting to $(\CH \otimes \CK_\infty)^o$ and with $n\to\infty$ we obtain the result.
\end{proof}

Note that
if $d = dim\, \CP \ge 2$ then $dim\, \CE_* = dim \CY = d-1 \ge 1$
and the scattering system $\big( (\CH \otimes \CK_\infty)^o, \underline{V} = (V_1,\ldots,V_d), \CG^+_*, \CG \big)$ is never trivial.
\\

In any outgoing Cuntz scattering system it is interesting to examine the relative position of the two embedded row shifts. In our case, having identified the wandering subspace $\CE_*$ as $W^* \CY$, it is more convenient to work instead with an equivalent row shift with wandering subspace $\CY$. The following considerations help to keep track of the relative position and will be used in the analysis of this problem in the following sections.

\begin{proposition} \label{prop:S}
Let $\underline{\widehat{S}} = (\widehat{S}_1,\ldots,\widehat{S}_d)$ with
\begin{eqnarray*} 
\widehat{S}_j: \;\CP_\infty &\rightarrow& \CP_\infty \\
\zeta &\mapsto& \epsilon_j \otimes \zeta \quad 
(\in \CP_1 \otimes \CP_{[2,\infty)})
\end{eqnarray*}
Then for $j = 1,\ldots,d$
\[
\widehat{S}_j\, \widehat{W} = \widehat{W} \, \widehat{V}_j
\]
If $\underline{S} = (S_1,\ldots,S_d)$ is the restriction to
$(\CP_\infty)^o$ then
\[
S_j\, W = W \,V_j
\]
\end{proposition}

\begin{proof}
Note that for $\zeta_\ell \in \CP_\ell$
\[
(\widehat{S}_j)^* \; \zeta_1 \otimes \zeta_2 \otimes \ldots
= \langle \epsilon_j, \zeta_1 \rangle \, \zeta_2 \otimes \ldots
\]
and hence we obtain for $\xi \in \CH,\; \eta \in \CK_\infty, \zeta_i \in \CP_i$
\begin{eqnarray*}
& & \langle \; (\widehat{W} \widehat{V}_j)\, \xi \otimes \eta,\;
\zeta_1 \otimes \zeta_2 \otimes \ldots \otimes \zeta_n \otimes \Omega_{[n+1,\infty)}^\CP
\; \rangle \\
&=& \langle \; \widehat{V}_j\, \xi \otimes \eta,\;
\widehat{W}^* \zeta_1 \otimes \zeta_2 \otimes  \ldots \otimes \zeta_n \otimes \Omega_{[n+1,\infty)}^\CP
\; \rangle \\
&=& \langle \; U^*(\xi \otimes \epsilon_j) \otimes \eta,\;
U_1^* \ldots U_n^*\, \OH \otimes \zeta_1 \otimes \zeta_2 \otimes \ldots \otimes \zeta_n \otimes \Omega_{[n+1,\infty)}^\CK
\; \rangle \\
&=& \langle \; \xi \otimes \epsilon_j \otimes \eta,\;
U_2^* \ldots U_n^*\, \OH \otimes \zeta_1 \otimes \zeta_2 \otimes \ldots \otimes \zeta_n \otimes \Omega_{[n+1,\infty)}^\CK
\; \rangle \\
&=& \langle \epsilon_j, \zeta_1 \rangle \;
\langle \; \xi \otimes \eta,\;
\widehat{W}^* \zeta_2 \otimes \ldots \otimes \zeta_n \otimes \Omega_{[n,\infty)}^\CP
\; \rangle \\
&=& \langle \; (\widehat{S}_j \widehat{W})\, \xi \otimes \eta,\;
\zeta_1 \otimes \zeta_2 \ldots \zeta_n \otimes \Omega_{[n+1,\infty)}^\CP
\; \rangle
\end{eqnarray*}
The intertwining relation for the restrictions is a direct consequence.
\end{proof}

\begin{corollary} \label{cor:equi}
$\underline{\widehat{S}}$ is a row unitary but $\underline{S}$
is a row shift with wandering subspace 
$(\Omega_1^\CP)^\perp \otimes \Omega_{[2,\infty)}^\CP \simeq \CY$. 
It is equivalent to the row shift on $\CG_*^+$ in Theorem \ref{thm:cuntz} via the intertwiner $W$. In particular
\[
\CE_* = W^* \CY, \quad W\,\CE_* = \CY, \quad\quad
\CG_*^+ = W^* (\CP_\infty)^o, \quad W \CG_*^+ =  (\CP_\infty)^o
\]
\end{corollary}

\section{$F^+_d$-linear systems and transfer functions}
\label{section:transfer}

For our model of an interaction $U$ with vacuum vectors $\OH,\, \OK, \,\OP$ we now want to study certain generalizations of linear systems theory which turn out to be closely connected to the outgoing Cuntz scattering system
$\big( (\CH \otimes \CK_\infty)^o,\, \underline{V} = (V_1,\ldots,V_d),\, \CG^+_*,\, \CG \big)$
which we have constructed in the last section.
We define
\begin{itemize}
\item[]
{\it the input space} 
$\quad \quad \CU := \CE = \CH \otimes (\Omega_1^\CK)^\perp \otimes \Omega_{[2,\infty)}^\CK \quad \subset (\CH \otimes \CK_\infty)^o$,

\vspace{0.1cm}

\item[]
{\it the output space}
$\quad \; \CY := (\Omega_1^\CP)^\perp \otimes \Omega_{[2,\infty)}^\CP
\quad \subset (\CP_\infty)^o$
\end{itemize}

With $H \otimes \CK = \CH \oplus \CU$ the interaction $U$ maps
$\CH \oplus \CU$ onto $\CH \otimes \CP$ which contains $\CY$
(identifying $\CP$ and $\CP_1$). Hence for $j=1,\ldots,d$ we can define 
\[
A_j: \CH \rightarrow \CH,\quad
B_j: \CU \rightarrow \CH,\quad
C  : \CH \rightarrow \CY,\quad
D  : \CU \rightarrow \CY
\]
by 
\begin{eqnarray*}
U (\xi \oplus \eta) &=:& \sum^d_{j=1} \big(A_j \xi + B_j \eta \big) \otimes \epsilon_j \\
P_\CY\, U (\xi \oplus \eta) &=:& C \xi + D \eta,
\end{eqnarray*}
where $\xi \in \CH,\, \eta \in \CU$ and $\big(\epsilon_j\big)^d_{j=1}$
is an orthonormal basis of $\CP$ and $P_\CY$ is the orthogonal projection onto $\CY$ (such a notation will also be used for other subspaces in the following). 

\begin{remark} \normalfont \label{rem:eigen}
Note that we found $A_1, \ldots, A_d \in \BH$ earlier when we considered the transition operator 
$Z(\cdot) = \sum_j A^*_j \cdot A_j: \BH \rightarrow \BH$ of the non-commutative Markov chain. The setting is not so special as it may seem on first glance. In fact, if
$(A^*_1,\ldots,A^*_d): \bigoplus^d_{j=1} \CH \rightarrow \CH$
is an arbitrary row contraction then we can use dilation theory to construct an isometric dilation $(V_1,\ldots,V_d)$ (see \cite{Po89a} for this kind of dilation theory) and to construct Hilbert spaces $\CK$ and $\CP$ and a unitary $U: \CH \otimes \CK \rightarrow \CH \otimes \CP$ which is related to $A_1,\ldots,A_d$ as above. See \cite{Go04} or \cite{DG07}, Section 1, where such a coisometry $U$ is explicitly constructed which can be extended to a unitary on enlarged spaces. It is a consequence of the existence of vacuum vectors in Definition \ref{def:interaction} that $\OH$ is a common eigenvector
for the tuple $A_1,\ldots,A_d$ (see \cite{Go04}, A.5.1) and hence, conversely,
if one starts with the tuple one has to assume this property in order to arrive at the setting of Definition \ref{def:interaction}.
\end{remark}

Further we define
\[
\CC_U :=
\left( \begin{array}{cc}
        A_1 & B_1 \\ 
        \vdots & \vdots \\
        A_d & B_d \\ 
        C & D
        \end{array} 
\right): 
\quad
\CH \oplus \CU \rightarrow \bigoplus^d_{j=1} \CH \oplus \CY
\]
which is called a colligation (of operators). 

As usual, the colligation $\CC_U$ gives rise to a {\it $F^+_d$-linear system} $\Sigma_U$ (also called a non-commutative Fornasini-Marchesini system in \cite{BGM06}, referring to \cite{FM78}), given by
\begin{eqnarray*}
x(j\alpha) &=& A_j \,x(\alpha) + B_j \, u(\alpha) \\
y(\alpha)  &=&  \; C \,x(\alpha)\, + \,  D \, u(\alpha),
\end{eqnarray*}
where $j=1,\ldots,d$, further $\alpha,\,j \alpha$ (concatenation) are words
in $F^+_d$ and 
\[
x: F^+_d \rightarrow \CH,\quad
u: F^+_d \rightarrow \CU,\quad
y: F^+_d \rightarrow \CY. 
\]
Given $x(\emptyset)$ and $u$ we can use $\Sigma_U$ to compute $x$ and $y$ recursively.

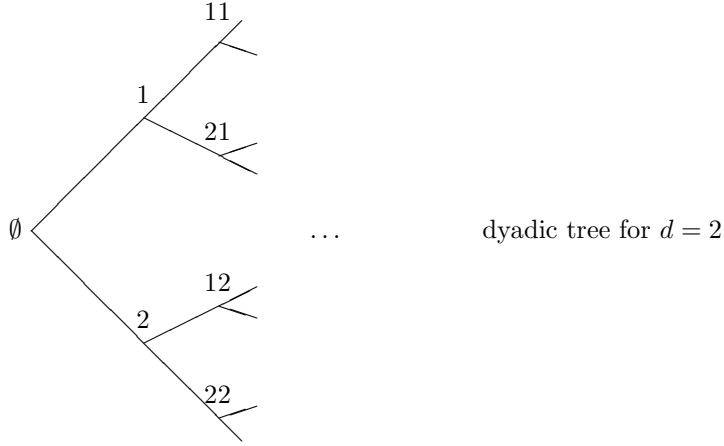
\setlength{\unitlength}{1cm}
\begin{picture}(15,6.5)
\put(2,3.2){\line(1,1){2.8}}
\put(2,3.2){\line(1,-1){2.8}}
\put(3.5,4.7){\line(2,-1){1.5}}
\put(3.5,1.7){\line(2,1){1.5}}
\put(4.5,5.7){\line(3,-1){0.5}}
\put(4.5,4.2){\line(3,1){0.5}}
\put(4.5,2.2){\line(3,-1){0.5}}
\put(4.5,0.7){\line(3,1){0.5}}
\put(1.7,3.1){$\emptyset$}
\put(3.4,4.9){$1$}
\put(3.4,1.9){$2$}
\put(4.3,6){$11$}
\put(4.3,4.4){$21$}
\put(4.3,2.4){$12$}
\put(4.3,0.9){$22$}
\put(5.7,3.1){$\ldots$}
\put(8,3.1){dyadic tree for $d=2$}
\end{picture}

A very elegant way to encode all the information about the evolution of an $F^+_d$-linear system into a single mathematical object is the use of a transfer function. For this we define the `Fourier transform' of $x$ as
\[
\hat{x}(z) = \sum_{\alpha\in F^+_d} x(\alpha) z^\alpha,
\]
where $z^\alpha = z_{\alpha_n} \ldots z_{\alpha_1}$ if $\alpha =
\alpha_n \ldots \alpha_1 \in F^+_d$ and $z = (z_1,\ldots,z_d)$
is a $d$-tuple of formal non-commuting indeterminates. Similarly
$\hat{u}(z) = \sum_{\alpha\in F^+_d} u(\alpha) z^\alpha$ and
$\hat{y}(z) = \sum_{\alpha\in F^+_d} y(\alpha) z^\alpha$.

Then it is easy to check that, if $x(\emptyset) = 0$, we have the
{\it input-output relation}
\[
\hat{y}(z) = \Theta_U(z)\, \hat{u}(z)
\]
where
\[
\Theta_U(z) := \sum_{\alpha \in F^+_d} \Theta_U^{(\alpha)} z^\alpha :=
D \, + \, C \sum_{\stackrel{\beta \in F^+_d}{j=1,\ldots,d}}
A_\beta B_j z^{\beta j} 
\]
The $\Theta_U^{(\alpha)}$ are operators from $\CU$ to $\CY$. Multiplication of the $z$-variables is done by concatenation of exponents and the coefficients are always assumed to commute with the $z$-variables.
We call the formal non-commutative power series $\Theta_U$ the {\it transfer function} associated to the interaction $U$. 

Now we want to proceed from formal power series to operators between Hilbert spaces.

\begin{theorem} \label{thm:M}
The input-output relation
\[
\hat{y}(z) = \Theta_U(z)\, \hat{u}(z)
\]
corresponds to a contraction
\[
M_{\Theta_U}: \ell^2(F^+_d,\CU) \rightarrow \ell^2(F^+_d,\CY)
\]
which (with $x(\emptyset) = 0$) maps an input sequence
$u$ to the corresponding output sequence $y$. 
\end{theorem}

\begin{proof}
We give a proof with colligations which illustrates the connections with the outgoing Cuntz scattering system constructed in the previous section.
One would like to use the colligation $\CC_U$ but one quickly observes that $\CC_U$ is in general not contractive. In fact, the colligation obtained by removing the last row $(C,\,D)$ is already unitary. We note however 
that the state variables $x(\alpha) \in \CH$ can be changed arbitrarily by scalar multiples of $\OH$ without changing the input-output map (because $y(\alpha) = P_\CY \,U \big(x(\alpha)+u(\alpha)\big)$ and $P_\CY\, U \big(\OH \otimes \OK \big) = 0$). Hence we can replace the colligation $\CC_U$ by the colligation
\[
\CC^o_U :=
\left( \begin{array}{cc}
        A_1^o & B_1^o \\ 
        \vdots & \vdots \\
        A_d^o & B_d^o \\ 
        C^o & D
        \end{array} 
\right): 
\quad
\CH^o \oplus \CU \rightarrow \bigoplus^d_{j=1} \CH^o \oplus \CY
\]
where $\CH^o := \CH \ominus \Cset \OH$ and 
\[
A_j^o: \CH^o \rightarrow \CH^o,\quad
B_j^o: \CU \rightarrow \CH^o,\quad
C^o  : \CH^o \rightarrow \CY,\quad
D    : \CU \rightarrow \CY
\]
are restrictions resp. compressions of $A_j,\,B_j,\,C,\,D$. 

Now recall from \cite{BV05}, Chapter 5.2, that given an outgoing Cuntz scattering system $\big( (\CH \otimes \CK_\infty)^o,\, \underline{V} = (V_1,\ldots,V_d),\, \CG^+_*,\, \CG \big)$ such that $(\CH \otimes \CK_\infty)^o = \CH^o \oplus \CG$ (as constructed by us in Theorem \ref{thm:cuntz})
we can associate a unitary colligation
\[
\left( \begin{array}{cc}
        \tilde{A}_1 & \tilde{B}_1 \\ 
        \vdots & \vdots \\
        \tilde{A}_d & \tilde{B}_d \\ 
        \tilde{C} & \tilde{D}
        \end{array} 
\right): 
\quad
\CH^o \oplus \CE \rightarrow \bigoplus^d_{j=1} \CH^o \oplus \CE_*
\]
by $(\tilde{A}_j,\,\tilde{B}_j) = P_{\CH^o} V^*_j |_{\CH^o \oplus \CE},\;
(\tilde{C},\, \tilde{D}) = P_{\CE_*} |_{\CH^o \oplus \CE}$. 
Unitarity follows directly from the geometry of the outgoing Cuntz scattering system, see \cite{BV05}.

Now observe that
$(A_j^o,\, B_j^o) = P_{\CH^o \otimes \epsilon_j} U |_{\CH^o \oplus \CE}$ (identifying $\CH^o$ and $\CH^o \otimes \epsilon_j$)
and $(C^o,\, D) = P_{\CY} U |_{\CH^o \oplus \CE}$ and hence
\begin{eqnarray*}
U^* (A_j^o,\, B_j^o) &=& U^* P_{\CH^o \otimes \epsilon_j} U |_{\CH^o \oplus \CE} = P_{U^* \CH^o \otimes \epsilon_j} |_{\CH^o \oplus \CE} \\
&=& P_{V_j \CH^o} |_{\CH^o \oplus \CE} = V_j P_{\CH^o} V^*_j |_{\CH^o \oplus \CE}
= V_j (\tilde{A}_j,\,\tilde{B}_j) \\
U^* (C^o,\, D) &=& U^* P_{\CY} U |_{\CH^o \oplus \CE}
= P_{U^* \CY} |_{\CH^o \oplus \CE} = (\tilde{C},\, \tilde{D})
\end{eqnarray*}
because $\CE_* = W^* \CY$, by Theorem \ref{thm:cuntz}, which with the identifications used here is the same as $U^* \CY$. 
It follows that the colligation $\CC_U^o$ is also unitary. This implies that $M_{\Theta_U}$ is contractive, by summing over all $\alpha$ the equations
\[
\sum^d_{j=1} |x^o(j \alpha)|^2 - |x^o(\alpha)|^2 
= |u(\alpha|^2 - |y(\alpha)|^2
\]
(with $x^o(\emptyset) = 0$), which can be
obtained from the unitary colligation $\CC_U^o$ replacing $\CC_U$.
\end{proof}

The operator $M_{\Theta_U}$ has the property that it intertwines with right translation, i.e., for all $j=1,\ldots,d$
\[
M_{\Theta_U} \big(\sum_{\alpha\in F^+_d} x(\alpha) z^\alpha\, z^j \big)
\;=\; M_{\Theta_U} \big(\sum_{\alpha\in F^+_d} x(\alpha) z^\alpha\big) \,z^j\;.
\]
Such operators have been called {\it analytic intertwining operators} in \cite{BV05} and {\it multi-analytic operators} in \cite{Po95} which refers to the fact that in the theory of these operators there are many analogues to the theory of multiplication operators by analytic functions on Hardy spaces. 
The non-commutative power series $\Theta_U$ is called the {\it symbol} of $M_{\Theta_U}$.
As discussed in the introduction it was one of the motivations for this paper to make this theory available for the study of interaction models and non-commutative Markov chains. Note further that because $M_{\Theta_U}$ is a contraction the transfer function $\Theta_U$ belongs to what in \cite{BV05}, 2.4, is called the {\it non-commutative Schur class} $S_{nc,d}\,(\CU,\CY)$. To compare with other work in the literature (for example \cite{Po95,DG07,DG}) we mention that $\ell^2(F^+_d,\CU)$ is naturally identified with a free Fock space tensored with $\CU$ in which case `multi-analytic' refers to intertwining with creation operators. While this is a very useful way to think about it we won't use it in this paper but write our formulas with the indeterminates $z$. 

To understand what the transfer function $\Theta_U$ can tell us about our physical model of interaction we construct a more explicit dictionary between the multiplicative tensor product description and the additive language of $F^+_d$-linear systems. 

We can interpret $\big( z^\alpha \big)_{\alpha \in F^+_d}$
as an orthonormal basis of $\ell^2(F^+_d, \Cset)$
and $\sum_{\alpha\in F^+_d} y(\alpha) z^\alpha$ for
square integrable coefficients $y(\alpha) \in \CY$
as a series converging to an element of $\ell^2(F^+_d, \CY)$.
It is natural to map $\zeta \in \CY$ to $\zeta z^\emptyset \in \ell^2(F^+_d, \CY)$ and, based on Corollary \ref{cor:equi}, to extend this to a unitary operator 
\begin{eqnarray*}
\Gamma_\CP: \;(\CP_\infty)^o &\rightarrow& \ell^2(F^+_d, \CY) \\
S_\alpha \zeta &\mapsto& \zeta\, z^{\overline{\alpha}}\;,
\end{eqnarray*}
where $\zeta \in \CY$ and $\overline{\alpha} = \alpha_1 \ldots \alpha_n$
is the reverse of $\alpha = \alpha_n \ldots \alpha_1 \in F^+_d$.
We have an intertwining relation
\[
\Gamma_\CP\,\big(S_\alpha\,\zeta\big)  \;=\; \big(\Gamma_\CP \,\zeta\big) z^{\overline{\alpha}}\;.
\]

Similarly, based on Theorem \ref{thm:cuntz}, we can define a unitary operator 
\begin{eqnarray*}
\Gamma_\CK: \;(\CH \otimes \CK_\infty)^o = \CH^o \oplus \CG
&\rightarrow& \CH^o \oplus \ell^2(F^+_d, \CU) \\
\xi \oplus V_\alpha \eta &\mapsto& \xi \oplus \eta\, z^{\overline{\alpha}}\;,
\end{eqnarray*}
where $\xi \in \CH^o,\,\eta \in \CU$.
Here the intertwining relation is
\[
\Gamma_\CK\,\big(V_\alpha\,\eta\big)  \;=\; \big(\Gamma_\CK \,\eta \big) z^{\overline{\alpha}}\;.
\]
 
By including the coisometry $W$ from Proposition \ref{prop:W} we can form the following useful commuting diagram.

\begin{theorem} \label{thm:comm}
Let $\Gamma_{W}$ be defined by the following commutative diagram:
\[
\xymatrix{
(\CH \otimes \CK_\infty)^o \ar[r]^{W}
\ar[d]_{\Gamma_\CK}
& (\CP_\infty)^o \ar[d]^{\Gamma_\CP}
\\
\CH^o \oplus \ell^2(F^+_d, \CU) \ar[r]^{\Gamma_W}
& \ell^2(F^+_d, \CY),
}
\]
i.e., $\Gamma_{W} = \Gamma_\CP\, W\, \Gamma_\CK^{-1}$.
Then we have
\[
\Gamma_{W} |_{\ell^2(F^+_d, \CU)} = M_{\Theta_U}
\]
\end{theorem}

\begin{proof}
Combining the intertwining relations of $\Gamma_\CK$ and $\Gamma_\CP$
with the intertwining relation 
$S_j\, W = W \, V_j$
from Proposition \ref{prop:S} we obtain for $\eta \in \CU,\,\beta \in F^+_d,\,j = 1,\ldots,d $
\[
\Gamma_{W} \big(\eta\, z^\beta z^j \big) = \Gamma_\CP\, W\, \Gamma_\CK^{-1} \big(\eta\, z^\beta z^j\big) = \Gamma_\CP\, W\,V_j\,V_{\overline{\beta}}\, \eta 
\]
\[
= \Gamma_\CP\,S_j\,S_{\overline{\beta}}\, W \eta
= (\Gamma_\CP\,W \eta)\, z^\beta z^j = \Gamma_{W} \big(\eta\, z^\beta \big)\, z^j 
\]
and we conclude that $\Gamma_W |_{\ell^2(F^+_d, \CU)}$ is a multi-analytic operator. To find its symbol it is enough to compute
$\Gamma_W \eta$ for $\eta \in \CU$ identified with $\eta\, z^\emptyset \in \ell^2(F^+_d, \CU)$.
For $\alpha = \alpha_{n-1} \ldots \alpha_1 \in F^+_d$, so that
$n = |\alpha|+1 \ge 1$ (which for $n=1$ means $\alpha = \emptyset$),
let $P_\alpha$ be the orthogonal projection
onto
\[
\Gamma_\CP^{-1} \{ \; f \in \ell^2(F^+_d, \CY): f= \zeta z^\alpha \;
{\mbox for\, some}\; \zeta \in \CY \} = S_{\overline{\alpha}}\, \CY
\]
\[
= \epsilon_{\alpha_1} \otimes \ldots \otimes \epsilon_{\alpha_{n-1}}
\otimes (\Omega_n^\CP)^\perp \otimes \Omega_{[n+1,\infty)}
\]
(which for $n=1$ is $\CY$). 
From Proposition \ref{prop:W} we obtain
\[
P_\alpha \, W\, \eta = P_\alpha \, U(n) \eta 
= P_\alpha \, U_n \ldots U_1 \eta 
\]
and now we can explicitly compute with the associated $F^+_d$-linear system $\Sigma_U$ from Section \ref{section:transfer}
\[
P_\alpha \, U_n \ldots U_1 \eta \quad = \quad
\left\{ \begin{array}{l@{\quad {\mbox if} \quad}l}
D \eta & n=1,\, \alpha = \emptyset \\
C A_{\alpha_{n-1}} \ldots A_{\alpha_2} B_{\alpha_1} \eta
& n = |\alpha| + 1 \ge 2
\end{array} \right.
\]
But this is exactly $\Theta_U^{(\alpha)} \eta$ for the coefficient $\Theta_U^{(\alpha)}$
of the transfer function $\Theta_U$, hence
$\Gamma_W |_{\ell^2(F^+_d, \CU)} = M_{\Theta_U}$.
\end{proof}

Clearly this commuting diagram provides another proof for Theorem \ref{thm:M}.
In the rest of this section we want to use our insights 
to give a very direct interpretation of what the transfer function means in a physical model of repeated interaction based on the axioms in Definition \ref{def:interaction}. 
As in Section \ref{section:interaction} we think of
$\CH$ as the (quantum mechanical) Hilbert space of an atom, $\CK_\ell$ as the Hilbert space of a part of a light beam or field which interacts with the atom at time $\ell$. We take $\CK = \CP$ (though it makes sense to distinguish mentally between $\CK_\ell$ as the Hilbert space at time $\ell$ immediately before and $\CP_\ell$ as the Hilbert space at time $\ell$ immediately after the interaction). 
We think of $\OH$ as a vacuum state of the atom and of $\OK = \OP$ in $\CK = \CP$ as a state indicating that no photon is present. 
Then $\eta \in \CU = \CH \otimes (\Omega_1^\CK)^\perp \otimes \Omega_{[2,\infty)}^\CK \; \subset \CH \otimes \CK_\infty$ represents a vector state with photons arriving at time $1$ and stimulating an interaction between the atom and the field, but no further photons arriving at later times. Nevertheless it may happen that some activity (emission) is induced which goes on for a longer period. To describe it quantitatively consider again for $\alpha = \alpha_{n-1} \ldots \alpha_1 \in F^+_d$ the projection
$P_\alpha$ which is the orthogonal projection
onto
\[
\epsilon_{\alpha_1} \otimes \ldots \otimes \epsilon_{\alpha_{n-1}}
\otimes (\Omega_n^\CP)^\perp \otimes \Omega_{[n+1,\infty)},
\]
already introduced in the proof of Theorem \ref{thm:comm}. In our physical interpretation it corresponds to the following event: We measure
data $\alpha_1, \ldots, \alpha_{n-1}$ at times $1,\ldots,n-1$ in the field, finally there is a last detection of photons corresponding to $(\Omega_n^\CP)^\perp$ at time $n$, nothing happens after time $n$.
This experimental record is obtained by measuring (at times indexed by the positive integers) an observable
$Y \in \BP$ with eigenvectors $\epsilon_1,\ldots,\epsilon_d$. Such lists of data have been used for indirect measurements of an atom, for quantum filtering and for updating protocols such as quantum trajectories, see for example \cite{BHJ,KM06} for such work with discrete time models. In our case the formula
\[
P_\alpha \, U(n) \eta \,=\, \Theta_U^{(\alpha)} \eta
\]
obtained in Theorem \ref{thm:comm} shows, according to the usual probabilistic interpretation of quantum mechanics, that
\[
\pi_\alpha := \|\, \Theta_U^{(\alpha)} \eta \,\|^2
\]
is the probability for the event described by $P_\alpha$ if we start in the state $\eta$ at time $0$. Actually the transfer function also keeps track of the complex amplitudes and contains additional coherent information.
This means that we can think of the transfer function $\Theta_U$ as a convenient way to assemble such data into a single mathematical object.

\section{Observability and scattering} \label{section:observ}

The identification of an outgoing Cuntz scattering system in Theorem \ref{thm:cuntz} allows to introduce many familiar concepts from linear systems theory. We single out a natural notion of observability and discuss its relations with ideas about a scattering theory of non-commutative Markov chains. 

\begin{definition} \label{def:observ}
The operator
\begin{eqnarray*}
W_0 := \Gamma_W |_{\CH^o}: \;\;\CH^o & \rightarrow & \ell^2(F^+_d, \CY) \\
         \xi & \mapsto &     \big( C A_\alpha \xi \big)_{\alpha \in F^+_d}
\end{eqnarray*}
is called the observability operator. The $F^+_d$-linear system is called (uniformly) observable if there exist $m,M >0$ such that
for all $\xi \in \CH^o$
\[
m\, \|\xi\|^2 \;\le\; \sum_{\alpha\in F^+_d} \| C A_\alpha \xi \|^2
= \|\,W_0 \xi\,\|^2
\le M \, \|\xi\|^2 .
\]
\end{definition}

Similar operators are also called Poisson kernels and interpreted accordingly by Popescu in \cite{Po99}. Note that for $dim\, \CH < \infty$
observability means that the observability operator $W_0$ is injective. This has a direct interpretation in our model which motivates the terminology quite well in this setting: Observability means that in an experiment with no $\CU$-inputs we can, by determining all $\CY$-outputs at all times, reconstruct the original state $\xi$ of the atom. We took the idea of uniform observability from \cite{FFGK97} where it is argued that this is a mathematically convenient generalization of observability for
$dim\, \CH = \infty$. We will see in the next theorem that this is also the case in our model. Note that 
we can always choose $M=1$ (by Theorem \ref{thm:comm}), so the crucial point for uniform observability is the existence of the uniform lower bound $m > 0$.

It is sometimes useful to extend $W_0$ to 
\[
\widehat{W}_0: \CH \rightarrow \Cset \oplus \ell^2(F^+_d, \CY) 
\]
by setting $\widehat{W}_0\, \OH := 1 \in \Cset$. If $W_0$ is uniformly observable then the defining inequalities extend to $\widehat{W}_0$ on $\CH$. 

\begin{theorem} \label{thm:scat}
The following assertions are equivalent:
\begin{itemize}
\item[(a)]
The system is (uniformly) observable.
\item[(b)]
The observability operator $W_0$ is isometric.
\item[(c)]
The transition operator $Z: \BH \rightarrow \BH$ is ergodic,\\
i.e., its fixed point space is $\Cset\,\eins$.
\item[(d)]
$W: (\CH \otimes \CK_\infty)^o \rightarrow (\CP_\infty)^o$
is unitary.
\end{itemize}
If these assertions are valid then we further have
\begin{itemize}
\item[(e)]
The transfer function $\Theta_U$ is inner, i.e., the multi-analytic operator $M_{\Theta_U}: \ell^2(F^+_d, \CU) \rightarrow \ell^2(F^+_d, \CY)$ is isometric. 
\end{itemize}
If $dim\,\CH < \infty$ and $dim\, \CP \ge 2$ then we also have the converse direction $(e) \Rightarrow (a,b,c,d)$.
\end{theorem}

\begin{proof}
$(d) \Rightarrow (b) \Rightarrow (a)$ is obvious. We now prove 
$(a) \Rightarrow (d)$, hence establishing the equivalence of $(a),(b)$
and $(d)$.

Given $0 \not= \eta \in \CH \otimes \CK_\infty$, approximate it by
$\eta^\prime \in \CH \otimes \CK_{[1,n]}$ (for some $n \in \Nset$)
such that 
\[
\| \,\eta - \eta^\prime \,\| \;<\; \frac{\sqrt{m}}{\sqrt{m}+1}\,\|\eta\|,
\]
here $m$ is the constant appearing in Definition \ref{def:observ} of (uniform) observability.

Suppose for a moment that $U_n \ldots U_1 \eta^\prime = \xi^\prime \otimes \zeta^\prime
\otimes \Omega_{[n+1,\infty)}^\CK$ with $\xi^\prime \in \CH$ and
$\zeta^\prime \in \CP_{[1,n]}$. Using Proposition \ref{prop:W} we find that $\|\,Q_N\, U_N \ldots U_n \ldots U_1 \eta^\prime \,\|$ tends for
$N\to\infty$ to $\|\widehat{W}\,\eta^\prime\,\|$.
But it also tends to
$\|\,\widehat{W}_0 \xi^\prime\,\|\,\|\,\zeta^\prime\,\|$ which by (uniform) observability is greater or equal than $\sqrt{m}\,\|\,\xi^\prime\,\|\,\|\,\zeta^\prime\,\|$. We conclude that
in this case we have $\|\,\widehat{W}\,\eta^\prime\,\|^2 \ge m\,\|\,\eta^\prime\,\|^2$. In the general case we can always write
$U_n \ldots U_1 \eta^\prime = \sum_j \xi^\prime_j \otimes \zeta^\prime_j
\otimes \Omega_{[n+1,\infty)}^\CK$ with $\xi_j^\prime \in \CH$ and
orthogonal vectors $\zeta^\prime_j \in \CP_{[1,n]}$. By handling the summands as above and summing up, we can verify the inequality above also in the general case, i.e., for all $\eta^\prime \in \CH \otimes \CK_{[1,n]}$
\[
\|\,\widehat{W}\eta^\prime\,\|^2 \;\ge\; m\,\|\eta^\prime\|^2 .
\]
Now we conclude that
\begin{eqnarray*}
\|\,\widehat{W}\eta\,\| 
&\ge& \|\,\widehat{W}\,\eta^\prime\,\| - \|\widehat{W}\,(\eta-\eta^\prime)\|\\
&\ge& \sqrt{m}\, \|\eta^\prime\| - \|\eta-\eta^\prime\| \\
&\ge& \sqrt{m}\, \|\eta\| - \big( \sqrt{m}+1 \big)\, \|\eta-\eta^\prime\| 
> 0 .
\end{eqnarray*}
Hence $\widehat{W}\eta \not= 0$ for all $0 \not= \eta \in \CH \otimes \CK_\infty$, i.e., $\widehat{W}$ is injective. But by Proposition \ref{prop:W} we also know that $\widehat{W}$ is a coisometry and an injective coisometry is unitary. Because $\widehat{W} \big(\OH \otimes \Omega_\infty^\CK\big) = \Omega_\infty^\CP$ it is clear that $W$ is unitary
if and only if $\widehat{W}$ is unitary.
This proves $(d)$.

To include condition $(c)$ we make use of the following well known facts about positive maps, see for example A.5.2 in \cite{Go04} for proofs. 
Because $Z: \BH \rightarrow \BH$ is positive with invariant vector state given by $\OH \in \CH$ it follows that $Z^n(p)$, with $p$ being the one-dimensional projection onto $\Cset \OH$, forms an increasing sequence which sot-converges to a fixed point $0 < x \le \eins$ of $Z$. The map $Z$ is ergodic if and only if $x=\eins$.

With the operators $Q_n$ introduced for Proposition \ref{prop:W}
we find for $\xi \in \CH$
\[
\|\, Q_n\,U(n) \xi\,\|^2 = \langle\,p \otimes \eins\;U(n) \xi,  \,p \otimes \eins\;U(n) \xi\,\rangle
= \langle\, \xi, U(n)^* p \otimes \eins\,U(n) \xi \,\rangle 
= \langle\, \xi, Z^n(p) \xi \,\rangle
\]
and hence
\[
\lim_{n\to\infty} \|\, Q_n\,U(n) \xi\,\|^2 = \langle\, \xi, x \xi \,\rangle,
\]
where $x$ is the fixed point of $Z$ mentioned above. But from Proposition \ref{prop:W} we also know that
$
\lim_{n\to\infty} \|\, Q_n\,U(n) \xi\,\| = \|\,\widehat{W} \xi\,\| = \|\,\widehat{W}_0 \xi\,\| .
$
and we obtain
\[
\|\,\widehat{W}_0\, \xi\,\|^2 = \langle\, \xi, x \xi \,\rangle .
\]
Using this it is easy to see that $(c)$ is equivalent to $(b)$. 
In fact, if we assume $(c)$, i.e. $Z$ is ergodic, then $x=\eins$ and hence we have $\|\,\widehat{W}_0\, \xi\,\| = \|\xi\|$ for all $\xi \in \CH$, which implies
$(b)$. If $Z$ is not ergodic then $x \not= \eins$ and hence there exists
a vector $\xi \in \CH$ so that
\[
\|\,\widehat{W}_0\, \xi\,\|^2 = \langle\, \xi, x \xi \,\rangle < \|\xi\|^2 ,
\]
which contradicts $(b)$. 

$(d) \Rightarrow (e)$ is clear because by Theorem \ref{thm:comm} the operator $M_{\Theta_U}$ is a restriction of $\Gamma_W$ which is unitarily equivalent to $W$. To consider the converse direction we define
\[
\CH_{scat} := \CH \cap \widehat{W}^* \CP_\infty
= \Cset \OH \oplus \{ \xi \in \CH^o:\; \|W_0\,\xi\| = \|\xi\|\; \}
\] 
From $\| \,\widehat{W}_0 \xi\,\| = \|\,\lim_{n\to\infty} Q\,U(n) \xi\,\|$ (Proposition \ref{prop:W})
we infer that
\[
U (\CH_{scat} \otimes \OK) \subset \CH_{scat} \otimes \CP .
\]
If we have $(e)$, i.e. $M_{\Theta_U} = \Gamma_W |_{\ell^2(F^+_d, \CU)}$ is isometric, then we must also have
\[
U (\CH \otimes (\OK)^\perp) \subset \CH_{scat} \otimes \CP .
\]
Together this implies that
\[
U^* \big( (\CH \ominus \CH_{scat}) \otimes \CP \big) 
\subset (\CH \ominus \CH_{scat}) \otimes \OK .
\]
If $dim\,\CH < \infty$ and $dim \CP \ge 2$ then a comparison of the dimensions of the left and right hand side of this inclusion forces
$\CH \ominus \CH_{scat} = \{0\}$,
i.e. $\CH_{scat} = \CH$, which is clearly equivalent to $(b)$. Hence under the stated assumptions about dimensions we also have the implication $(e) \Rightarrow (b)$.
\end{proof}  

Some of the techniques used in the proof of Theorem \ref{thm:scat} resemble those used in the scattering theory of non-commutative Markov chains, see \cite{KM00,Go04,GKL06}. This is not an accident. In fact, a discrete time version of this theory can be based on the following data: a unital $C^*$-algebra $\CA$ with a state $\phi$, another unital $C^*$-algebra $\CC$ with a state $\psi$ and an automorphism $\alpha$ of $\CA \otimes \CC$ which has $\phi \otimes \psi$ as an invariant state. This automorphism is interpreted as an interaction and gives rise to a one-step dynamics of a non-commutative Markov chain. We may consider this setting as a generalization of our model where we only considered the algebras of all operators on a Hilbert space together with pure states. But in a way we can also reduce the general setting to our model by using the GNS-construction: Let $\CH$ be the GNS-space of $\CA$
with cyclic vector $\OH$ and $\CK$ the GNS-space of $\CC$ with cyclic vector $\OK$. Define a unitary operator $U: \CH \otimes \OK \rightarrow \CH \otimes \OK$ by
\[
U^*: y\;\OH \otimes \OK \mapsto \alpha(y)\;\OH \otimes \OK
\]
(We take this to be $U^*$ instead of $U$ because the automorphism $\alpha$ acting on algebras of observables represents the Heisenberg picture while $U$ in this paper is designed to represent the Schr\"{o}dinger picture.) Then with $\CK=\CP$ and $\OK=\OP$ we have constructed the setting of our Definition \ref{def:interaction} and the results of this paper apply. In particular we may call the corresponding
transfer function $\Theta_U$ {\it the transfer function of the stationary state non-commutative Markov chain}. It is shown in \cite{Go04} that a substantial part of the scattering theory of these Markov chains can be studied and simplified on the level of the GNS-construction. In particular, as can be seen from comparing condition $(c)$ of Theorem \ref{thm:scat} with
\cite{Go04}, 2.6, or \cite{GKL06}, Th.4.3, observability in the sense of Definition \ref{def:observ} is equivalent to asymptotic completeness in the scattering theory of the Markov chain. The notation $\CH_{scat}$ introduced in the proof of Theorem \ref{thm:scat} is motivated by the idea of scattering states: In our model this means that in the long run they asymptotically end up in what we called $\OH \otimes \CP_\infty$ which is completely disentangled from the atom described by $\CH$. 
From this point of view we can look at the approach in this paper as a way to generalize the scattering theory of non-commutative Markov chains to situations which are not asymptotically complete. The transfer function is then closely related to (and generalizes) the Moeller operator in the terminology of \cite{KM00,Go04,GKL06}. 

Finally we want to comment about the relationship between this paper and the operator theoretic investigations about characteristic functions in \cite{DG07,DG}. Recall that we can write 
$Z(\cdot) = \sum_j A^*_j \cdot A_j: \BH \rightarrow \BH$
for the transition operator and $(A^*_1,\ldots,A^*_d)$ is a row contraction with $\OH$ a common eigenvector for $A_1,\ldots,A_d$,
say $A_j\, \OH = \omega_j\, \OH$ and $\omega_j \in \Cset,\; j=1,\ldots,d$
(compare Remark \ref{rem:eigen}).
With $\CH = \Cset\,\OH \oplus \CH^o$ and $A^o_j$ the compression of $A_j$ to $\CH^o$ we have $(A^o_j)^* = A^*_j |_{\CH^o}$ and we obtain block matrices of the form 
\[
A^*_j \;=\;
\left( \begin{array}{cc}
        \overline{\omega}_j & 0 \\ 
        * &  (A^o_j)^* \\ 
        \end{array} 
\right),\quad\quad j=1,\ldots,d
\]
i.e., $(A^*_1,\ldots,A^*_d)$ is a lifting in the terminology of \cite{DG}.
The characteristic function introduced in \cite{DG} is a multi-analytic operator associated to a lifting and the ergodic case is studied in detail in \cite{DG07}. As shown in \cite{DG07},
Prop. 2.1, another condition equivalent to $(a),\,(b),\,(c),\,(d)$
of our Theorem \ref{thm:scat} is
\begin{itemize}
\item[($c^\prime$)]
$(A^o_1,\ldots,A^o_d)$ is stable, i.e., for all $\xi \in \CH^o$
\[
\lim_{n\to\infty} \sum_{|\alpha|=n} \| A^o_\alpha\, \xi\|^2 = 0\;,
\]
\end{itemize}
which is a variant of the close connection between observability and stability well known in linear systems theory, see for example the discussion in Chapter III of \cite{FFGK97}. 

The main difference between \cite{DG07,DG} and the investigations presented here is as follows. In \cite{DG07,DG} we start from the tuple $\underline{A}$ or the map $Z$ and use the theory of minimal isometric dilations. The multi-analytic operator obtained is thus a characteristic function associated to $\underline{A}$ or $Z$. In this paper we do not consider minimality but we start with the interaction $U$ and obtain a multi-analytic operator which represents the transfer function of an input-output system associated to the interaction. This additional flexibility is very useful in the physical modeling because, unlike for classical Markov chains, there are quite different environments and interactions which give rise to the same transition operator $Z$, compare \cite{Ku03}. Hence we expect that the scheme developed here is more directly applicable to physical models. Of course the study of non-minimal dilations is also of interest for operator theory. On the other hand we note that in the setting of \cite{DG} the assumption of a one-dimensional eigenspace is dropped and the theory is much more general in another direction. A further integration of these schemes in the future may help to remove unnecessarily restrictive assumptions of the toy model considered in this paper and lead to the study of other and of more realistic models.

\end{document}